\pdfoutput=1

\documentclass[reqno]{amsart}

\usepackage[latin1]{inputenc}

\usepackage[section]{algorithm}
\usepackage{algorithmic}
\usepackage{microtype}
\usepackage{amssymb}
\usepackage{amsmath}
\usepackage{amscd}
\usepackage{amsthm}
\usepackage{amsfonts}
\usepackage{enumerate}
\usepackage{graphicx}
\usepackage{url}
\usepackage{amssymb}
\usepackage[dvips]{color}
\usepackage{epsfig}
\usepackage{mathrsfs}

\newtheorem{theorem}{Theorem}[section]

\newtheorem{openprob}{Open Problem}

\theoremstyle{definition}

\newtheorem{definition}[theorem]{Definition}

\theoremstyle{remark}

\numberwithin{equation}{section}

\date{}
\author{Erik Demaine, Nathan Pinsker and Jon Schneider}
\title{Fast Dynamic Pointer Following via Link-Cut Trees} 
\begin{document}
\maketitle

\begin{abstract}
In this paper, we study the problem of fast dynamic pointer following: given a directed graph $G$ where each vertex has outdegree $1$, efficiently support the operations of i) changing the outgoing edge of any vertex, and ii) find the vertex $k$ vertices `after' a given vertex. We exhibit a solution to this problem based on link-cut trees that requires $O(\lg n)$ time per operation, and prove that this is optimal in the cell-probe complexity model.
\end{abstract}

\section{Introduction}

Consider the following problem. Start with $N$ nodes, labelled $1$ through $N$. Each node contains a pointer to exactly one other node (possibly itself). The goal is to efficiently support the following two operations:

\begin{itemize}
\item
\textbf{Update}: Change the pointer of node $v$ to point to node $w$

\item
\textbf{Query}: Return the node $k$ pointers `ahead' of a node $v$. That is, if we define $f(v)$ to be the node that $v$ points to, we are concerned with computing the value of the $k$th iterate $f^{k}(v)$.
\end{itemize}

Naively, this is easy to do in $O(1)$ time per update and $O(k)$ time per query. In this paper, we present a data structure based on link-cut trees that can perform both of these operations in $O(\lg n)$ time. Moreover, we show that this data structure is optimal in the cell-probe complexity model.

Our paper is organized as follows. In Section \ref{pwork}, we describe some earlier work on the related problem of level ancestors in a tree and how it has previously been adapted to the dynamic case. In Section \ref{mres}, we provide the construction of our data structure and prove that it can perform both of the desired operations in $O(\lg n)$ time. In Section \ref{addops}, we consider additional operations our data structure can support, such as cycle-finding and least common ancestor queries. Finally, in Section \ref{opt}, we demonstrate how to apply a recent $\Omega(\lg n)$ lower bound for dynamic connectivity to show the optimality of our data structure. 

\section{Background}\label{pwork}

While (somewhat surprisingly) very little about this exact problem seems to have been studied, our problem is very similar in flavor to the problem of \textit{level ancestors}. The level ancestors problem asks, given a rooted tree $T$, a node $v$ of this tree, and an integer $d$, what is the $d$th ancestor of $v$ in $T$; that is, what ancestor do we reach by starting at $v$ and walking $d$ steps towards the root. Note that the fast dynamic pointer following problem can be seen as a type of level ancestors problem, with the exception that instead of a rooted tree, the underlying graph is a directed graph where each vertex has outdegree $1$. In fact, by directing each edge in a tree $T$ to point towards the root, we can view the dynamic level ancestors problem as a specific case of the dynamic pointer following problem. In Section \ref{mres} we will see that solving the level ancestor problem for fully dynamic trees plays an important role in solving the problem of fast dynamic pointer following. 

Variants of our problem in the static case have been extensively studied. For example, in \cite{bfc}, Bender and Farach-Colton present an optimal solution for the problem of level ancestors in a static tree that requires $O(n)$ space and $O(1)$ time per level ancestor (earlier constructions that matched this linear-space, constant-time bound were also known). Moreover, in \cite{perm}, Munro, Raman, Raman, and Rao demonstrate a succinct data structure for solving the static version of the pointer-following problem. 

For the case of dynamic trees, much less appears to be known. The results that do exist in this case tend to be aimed more towards specific restrictions on the types of operations we are allowed to perform on dynamic trees. For example, in \cite{lvl}, Dietz shows how to maintain linear space and constant time for queries and updates in the case where our update operation consists of adding new leaves and new roots to our tree. Similarly, Alstrup and Holm show how to obtain constant time for queries and updates in the case where only the addition of new leaves is allowed; they furthermore show how to obtain an amortized inverse Ackermann bound for queries and updates in the case where edges can be added between two disjoint trees in a forest \cite{lvl2}. 

\section{Main result and proof}\label{mres}

\subsection{Overview}

We will show in this section that the \textbf{update} and \textbf{query} operations for dynamic pointer following can both be performed in $O(\lg n)$ time.  We do this by first showing we can find any ancestor of a given node $v$ in a dynamic tree in $O(\lg n)$ time (the \textit{dynamic level ancestor} problem), and then extend this to solve the problem stated above.

\subsection{Fully dynamic level ancestors via link-cut trees}

We first turn our attention to the problem of finding \textit{level ancestors in a fully dynamic tree}.  Similarly as in our problem, we want to efficiently support the following operations:

\begin{itemize}
\item
\textbf{Update}: Change the parent pointer of node $v$ to point to node $w$ or NULL

\item
\textbf{Query}: Return the node at depth $d$ that is an ancestor of a node $v$.
\end{itemize}

To do this, we use the link-cut tree data structure, described by Sleator and Tarjan in \cite{link}.  Given a dynamic tree $d$, we build a link-cut tree by making a node $n'$ for each node $n$ in $d$.  We then use the \textit{link} operation to set the parent of $n'$ to the node in our link-cut tree corresponding to the parent of $n$ in $d$.

To \textbf{update} a node $v$ to point to a node $w$, we simply call \textit{cut(v)} followed by \textit{link(v, w)}.  This takes $O(\lg n)$ time.

To \textbf{query} a node at depth $d$ that is an ancestor of a node $v$, we first call \textit{access(v)}. This changes the structure of our link-cut tree in an advantageous way.  It guarantees that $v$ and all ancestors of $v$ will be contained in a single splay tree rooted at $v$, where all elements in this splay tree are keyed by their depth in the original tree.  This allows us to easily find our desired node in the splay tree: after calling \textit{access(v)}, we simply search the splay tree rooted at $v$ for the node with depth $d$ with a normal splay-tree $find$ operation.  Since the splay tree contains only the path from the root down to $v$, no two nodes in this splay tree have the same depth (key), and so this node is the (unique) ancestor of $v$ at depth $d$. The $access$ and $find$ operations both take $O(\lg n)$ time, so our query operation shares this time bound.

\subsection{Structure of pseudoforests}

In this section, we will briefly consider the properties of graphs formed by allowed sets of nodes and pointers. We call any directed graph where each node has exactly one outgoing edge a \textit{pseudoforest}; note that these graphs exactly specify the state of our nodes and pointers at any given time.

\begin{theorem} \label{cycle}
A connected component in a pseudoforest contains exactly one directed cycle. All edges that do not belong to the cycle of the connected component must belong to a tree rooted at a node belonging to the cycle, where all directed edges of the tree point towards the cycle.
\end{theorem}
\begin{proof}
Let $G$ be our pseudoforest, and consider an arbitrary connected component $H$ of $G$.  We first prove that $H$ contains a directed cycle.

Let $n = |H|$, and for any node $a \in H$, consider the set of nodes $\lbrace f(a), f^{2}(a), \dots,$ $f^{n+1}(a) \rbrace$.  By the Pigeonhole Principle, two of these nodes are equal; call these nodes $f^{r}(a)$ and $f^{s}(a)$, where $r < s$.  Then $f^{s-r}(f^{r}(a)) = f^{r}(a)$, implying the node $f^{r}(a)$ is part of a cycle.  Thus $H$ contains a directed cycle.
\\

We next want to prove that all other nodes `point into this cycle' -- in other words, for any node $a$ and a sufficiently large integer $k$, $f^{k}(a)$ belongs to this cycle.  By similar logic as above, we know for some large enough $k$,  $f^{k}(a)$ is part of some cycle.  Assume to the contrary that $f^{k}(a)$ is part of a different cycle than the cycle we previously found.  It still must be connected to the cycle, however, and since $a$ has only one out-pointer, we know $a$ must be connected to our cycle through a node $b$ such that $f(b) = a$.  However, we can apply the same logic to $b$, since $f^{k+1}(b) = f^{k}(a)$, and obtain another node $c$ such that $f(c) = b$.  This means that we can construct an infinite sequence of nodes in our connected component, which is an impossibility since our entire graph has $n$ nodes.  This contradicts our assumption that $f^{k}(a)$ is part of a different cycle, meaning all nodes in this component point towards a single cycle.
\end{proof}

From the proof above, we can see that if we designate a node in a cycle as a ``root'' node, then our graph is a tree save for a single edge: the edge pointing out of our root node, to some other node.  We will make use of this insight in the next section.

\subsection{Reduction to fully dynamic level ancestors}

We are now ready to solve the dynamic pointer following problem.  In the previous subsections, we demonstrated how to use link-cut trees to find level ancestors in a dynamic tree.  We will construct a data structure very similar to our fully dynamic level ancestors data structure, with one important difference: we will keep track of a separate pointer for the root node of each tree in our forest, which will point to one of its children. As we have just shown, this structure encodes all information that is needed for our problem.

Our first problem is how to efficiently implement queries; that is, how can we quickly calculate $f^{k}(a)$ for some node $a$ in our tree and some $k$. Note that if $k < depth(a)$, then $f^{k}(a)$ can be found using our fully dynamic level ancestors algorithm described above. Therefore, it suffices to solve this problem in the case where $k \geq depth(a)$.

Let $r$ be the root node of the tree, let $f(r) = b$, and let $d(a)$ be the depth of $a$ in this tree rooted at $r$; we let $d(r) = 0$. By the structure of our tree, $r$ and $b$ are part of a cycle of length $depth(b)$.  It follows that, for any other node $a$, if $k \geq depth(a)$ then $f^{k}(a)$ must belong to this cycle. Since $f^{k}(a) = f^{k - depth(a)}(b)$, and since $b$ is part of a cycle, $f^{depth(b)}(b) = b$.  This means $f^{k - depth(a)}(b)$ = $f^{(k - depth(a))\bmod b}(b)$.  Since $(k - depth(a))\bmod b < b$, we can answer this query using the level-ancestor methods described above.  Calculating $depth(a)$ and $depth(b)$ takes $O(\lg n)$ time, and calculating $f^{(k - depth(a))\bmod b}(b)$ takes $O(\lg n)$ time, meaning each query has an overall runtime of $O(\lg n)$ as well.

Updates are treated very similarly to those in our solution to dynamic level-ancestor queries.  If a node outside of the cycle in our tree is updated, then the cycle is completely unchanged and we can process the update as we did above.  However, if a node in the cycle is updated, then the cycle has potentially changed.  However, in all cases, the node has been updated must still be in the cycle.  This is because any cycle that has been formed as a result of the update must have the updated edge contained within it, and therefore must have the node that has been updated within it.

As a result, if any node in the tree other than the root is updated, that node will still be in the tree.  We can therefore designate that node as the new root.  The section of the tree pointed to by the new root will be cut and will be inserted as a child of the new root using the \textit{link} operation, and the new root's pointer will be updated to its appropriate child.  All of these operations take $O(\lg n)$, and so update takes $O(\lg n)$ as well.

\section{Additional operations}\label{addops}

In the previous section, we described how to perform the operations Update and Query in $O(\lg n)$ time. For various applications we may wish to support some additional operations; luckily, the data structure presented above is powerful enough to perform a wide range of other operations and queries about the structure of our graph in logarithmic time. We discuss below some of these other possible operations.

\begin{itemize}
\item
\textbf{Cycle length:}
Given a node $x$, we would like to return the length of the cycle of the connected component it belongs to. We saw how to do this in our implementation of Query in the previous section; if $r$ is the root of the component containing $x$ and $r$ points to a node $s$, then the cycle length is just $1$ greater than the depth of $s$. 
\item
\textbf{Cycle belonging:}
Does $x$ belong to the cycle of its connected component? If this is the case, then $x$ must be an ancestor of the node $s$ defined in the previous point; since we know the depths of $x$ and $s$, we can check for this with one query operation. 
\item
\textbf{Inverse Query:}
Given two nodes $x$ and $y$, we want to search for $k$ such that $Query(x, k) = y$ (or if no such $k$ exists, then output so). Note that if such a $k$ exists, then $y$ must either belong to the path from $x$ to the root or on the cycle (the path from $y$ to the root). Since we know the depths of $x$ and $y$, we can check both of these cases easily with query operations like in the previous point. Once we know which (or both of) these cases hold, we can then immediately compute $k$ given the depths of $x$, $y$, and the length of the cycle (the depth of $s$). 
\item
\textbf{Delete:}
If a node $x$ has no incoming edges, we want to be able to delete it. Since $x$ has no incoming edges, it must be a leaf in its associated link-cut tree, so it suffices to just delete this leaf from the link-cut tree (which we can do in $O(\lg n)$ time). 
\item
\textbf{Subdivide:}
Given a directed edge $e$ from $x$ to $f(x)$ in our graph, we would like to add an intermediate node $y$ such that the successor of $x$ is $y$ and the successor of $y$ is $f(x)$. This is easy to do with two update queries. 
\item
\textbf{`Least Common Ancestor':}
Given $x$ and $y$ belonging to the same component, return the `least common ancestor' of $x$ and $y$. Note that, since each connected component contains a cycle, this isn't always uniquely defined, so we will assume for now that if the least common ancestor of $x$ and $y$ lies on the cycle, we can return any point on the cycle. If this common ancestor doesn't lie on the cycle, then $x$ and $y$ must belong to the same tree in the forest obtained by deleting this cycle from the graph; in this case, the least common ancestor is well-defined.

Implementing this operation is slightly trickier than the implementation for the preceding operations, in that (as with query), we must consider the internal structure of the link-cut tree. Note that the least common ancestor of $x$ and $y$ in the link-cut tree is also a valid output for the least common ancestor of $x$ and $y$ in our graph, so it suffices to compute the `link-cut' least common ancestor of $x$ and $y$.

To do this, note that the least common ancestor of $x$ and $y$ is the deepest node that belongs to both the path from the root down to $x$ and the path from the root down to $y$. Now, recall that upon accessing a node $v$, the path from the root to $v$ becomes a preferred path. Therefore, if we first access $x$ and then access $y$, the preferred path starting at the root changes from being a path down to $x$ to being a path down to $y$. This means that at some point, the preferred child of the least common ancestor changes (from the previous ancestor of $x$ to the previous ancestor of $y$); moreover, of all preferred child changes, this one occurs the least deep in the tree. Therefore, by keeping track of the preferred child changes during the access of $y$, we can compute the least common ancestor of $x$ and $y$, as desired.

\item
\textbf{Cycle Proximity:}
Given a node $x$, what is the smallest $k$ for which $Query(x, k)$ belongs to the cycle of the connected component containing $x$ (equivalently, how many steps forward must we walk before we enter the cycle)? 

To answer this, we apply the least common ancestor algorithm from the previous point. Since the nodes forming the cycle of our component correspond exactly to the nodes on the path from the root down to the node $s$, the least common ancestor of $x$ and $s$ is the first ancestor of $x$ which belongs to this cycle. The value $k$ we seek is then simply the difference in depths of this node and $x$. 
\end{itemize}

One operation that we do not know how to perform in time $O(\lg n)$ is \textbf{Contract}; that is, given two adjacent nodes $v$ and $f(v)$, collapse them into a single node whose set of incoming edges is the union of the incoming edges of $v$ and $f(v)$, and whose outgoing edge is to the node $f(f(v))$. The difficulty in this case is that $v$ and $f(v)$ might both have a large number of incoming edges; merging the two nodes would then possibly require a large (possibly linear in $n$) number of update operations.

\begin{openprob}
Can we perform the \textbf{Contract} operation in time $O(\lg n)$?
\end{openprob}

\section{Optimality}\label{opt}

In this section, we prove that no data structure working in the cell-probe model \cite{cprobe} can solve the fast dynamic pointer following problem with time faster than $\Omega(\lg n)$ per operation. 

Our main tool for doing so is the following result due to Demaine and Patrascu \cite{conn}:

\begin{definition}
In the problem of \textit{dynamic connectivity}, we must support the following operations on a set of $n$ nodes: i) add a (possibly directed) edge between two nodes, ii) delete an already existing edge between two nodes, iii) query whether two nodes $v$ and $w$ belong to the same connected component. The problem of \textit{path dynamic connectivity} is the same as the problem of dynamic connectivity, except the connected components of of our graph are always all guaranteed to be directed paths (as in the path grid construction of \cite{conn}).
\end{definition}

\begin{theorem}\label{depa}
Under the cell-probe model, any algorithm that solves the problem of \textit{path dynamic connectivity} requires at least $\Omega(\log n)$ time per operation.
\end{theorem}

On the other hand, it turns out that any algorithm that can solve the fast dynamic pointer following problem can solve the path dynamic connectivity problem equally as fast:

\begin{theorem}\label{red}
Any algorithm that solves the dynamic pointer following problem in $O(t)$ time per operation can also solve the path dynamic connectivity problem in $O(t)$ time per operation.
\end{theorem}
\begin{proof}
We demonstrate how to simulate all the operations required for the path dynamic connectivity problem with a constant number of operations in the dynamic pointer following (DPF) problem model. 

Initially, assume that in our path dynamic connectivity instance we start with no edges, and in our DPF instance each node starts with a pointer to itself. To simulate addition of a directed edge from node $v$ to $w$, we first claim that node $v$ must currently point to itself in the DPF instance; if this is not the case, then this will imply that $v$ already has an outgoing edge, contradicting the condition that at each point in time the set of connected components is a collection of directed paths. Similarly, to simulate deletion of the directed edge from node $v$ to $w$, we simply set $v$ to point to itself. In this way, all the directed edges in our DPF instance that are not self-loops correspond with all the directed edges in our path dynamic connectivity instance. 

Now, note that if $v$ and $w$ belong to the same component (which is a directed path), then $Query(v, n)$ and $Query(w, n)$ must both be equal to the node at the end of this directed path (here $n$ is equal to the number of nodes in our instance). On the other hand, if $v$ and $w$ do not belong to the same component, it is impossible for $Query(v, n) = Query(w,n)$ (this would imply some path exists between $v$ and $w$). Therefore, to check for connectivity between $v$ and $w$, it suffices to check whether $Query(v, n) = Query(w, n)$, which requires only $2 = O(1)$ queries in our DPF model.
\end{proof}

\begin{theorem}
Under the cell-probe model, any algorithm that solves the fast dynamic pointer following problem must spend at least $\Omega(\lg n)$ time per operation.
\end{theorem}
\begin{proof}
This follows immediately from Theorems \ref{depa} and \ref{red}.
\end{proof}


\begin{thebibliography}{99}

\bibitem{lvl2} Alstrup, S., Holm, J.: Improved Algorithms for Finding Level Ancestors in Dynamic Trees. ICALP 00:73-84, 2000.

\bibitem{med} Alstrup, S., Holm, J., Thorup, M.: Maintaining Center and Median in Dynamic Trees. SWAT 00:46-56, 2000.

\bibitem{bfc} Bender, M. A., Farach-Colton, M: ``The Level Ancestor Problem Simplified.'' Theoretical Computer Science Special Issue on LATIN '02, 321(1):5-12, 2004.

\bibitem{conn} Demaine, E., Patrascu, M.: Lower Bounds for Dynamic Connectivity. \textit{ Proceedings of the 36th ACM Symposium on Theory of Computing} (STOC 2004), Chicago, Illinois, June 13-15, 2004, 546-553. 

\bibitem{lvl} Dietz, P. F.: Finding Level-Ancestors in Dynamic Trees. LNCS, 519:32-40, 1991.

\bibitem{cprobe} Gal, A., Miltersen, P.B.: The cell probe complexity of succinct data structures. Theor. Comp. Sci. 379(3): 405-417, 2007.

\bibitem{perm} Munro, I., Raman, R., Raman, V., Rao, S.: Succinct representations of permutations and functions. Theor. Comp. Sci. 438: 74-88, 2012.

\bibitem{link} Sleator, D.D., Tarjan, R.E.: A data structure for dynamic trees. JCSS 26, 362-391, 1983.

\end{thebibliography}
\end{document}